\documentclass[twocolumn,twoside,superscriptaddress,prl,aps,showpacs]{revtex4}
\usepackage{amsmath,mathrsfs,amsbsy,color,graphicx,bm,amsthm,amsfonts}
\usepackage{times}
\usepackage{units}

\begin{document}

\newcommand{\ket}[1]{|{#1}\rangle}
\newcommand{\gateset}{S}
\newtheorem{Proposition}{Proposition}

\title{Solovay-Kitaev Decomposition Strategy for Single-Qubit Channels}

\author{Dong-Sheng Wang}
\affiliation{Institute for Quantum Science and Technology, University of Calgary, Alberta T2N 1N4, Canada}
\author{Dominic W. Berry}
\affiliation{Department of Physics and Astronomy, Macquarie University, Sydney, New South Wales 2109, Australia}
\author{Marcos C. de Oliveira}
\affiliation{Institute for Quantum Science and Technology, University of Calgary, Alberta T2N 1N4, Canada}
\affiliation{Instituto de F\'{\i}sica Gleb Wataghin, Universidade Estadual de Campinas, 13083-859,
	Campinas, S\~{a}o Paulo, Brazil}
\author{Barry C. Sanders}
\affiliation{Institute for Quantum Science and Technology, University of Calgary, Alberta T2N 1N4, Canada}

\begin{abstract}
Inspired by the Solovay-Kitaev decomposition for approximating
unitary operations as a sequence of operations selected from a
universal quantum computing gate set, we introduce a method for approximating any single-qubit channel
using single-qubit gates and the controlled-\textsc{not} (\textsc{cnot}).
Our approach uses the decomposition of the single-qubit channel into a convex combination of ``quasiextreme'' channels.
Previous techniques for simulating general single-qubit channels would require as many as 20 \textsc{cnot} gates,
whereas ours only needs one, bringing it within the range of current experiments.

\end{abstract}

\maketitle

Quantum computing requires the capability to efficiently approximate arbitrary quantum operations
as a sequence of a finite set of operations.
The celebrated Solovay-Kitaev theorem~\cite{NC00,KSV02} addresses this problem
by providing a strategy for approximating any unitary operation~$U$ within error tolerance~$\epsilon$
as a sequence of $O({\rm polylog}(1/\epsilon))$ gates chosen from the finite set.
Dawson and Nielsen~\cite{DN06} introduced an algorithm for the Solovay-Kitaev decomposition,
and many improvements have appeared recently~\cite{KMM13a,KMM13b,TvMH13,DCS12,PS12,GS13,VK13}.
These algorithms are central to quantum simulation efforts,
which is especially important because quantum simulation is regarded as the most promising
direction for a nontrivial quantum computation~\cite{AL97}.

Closed-system (i.e., Hamiltonian-generated) quantum simulation is well established~\cite{Llo96,AT03,BACS07a,Chi09,WBHS11},
but open-system quantum simulation is still at an early stage with attention focused on
simulating memoryless (Markovian) dynamics based on a Lindblad master equation~\cite{BCC+01,LV01,WML+10,WAN11,KBG+11,BK12}.
Open-system quantum simulation is important to cool to the ground state~\cite{ADLH+05},
prepare thermal states~\cite{TD00,PW09}
and entangled states~\cite{BMS+11,MDPZ12}, and
study nonequilibrium quantum phase transitions~\cite{PI11}.
Conversely, dissipative dynamics can be a resource for universal quantum computing \cite{VWC09}.

Given the importance of open-system quantum simulation,
efficiently approximating channels rather than just approximating unitary evolution is critical.
Here, we solve single-qubit channel simulation,
developing methods that could ultimately be adapted for multiqubit channels.
Our channel simulator could be regarded as a primitive for simulating open-system dynamics,
in the same way as single-qubit unitary gates are a primitive for closed-system dynamics.

An obvious direction for implementing a channel is applying Stinespring dilation
to implement a channel as a unitary operator on an expanded Hilbert space.
This resultant unitary transformation can then be implemented by standard techniques~\cite{KMM13a,KMM13b,TvMH13,DCS12,PS12,GS13,VK13}.
The problem with this approach is that it requires implementing a general unitary operator on a space with dimension given by the cube of the Hilbert space dimension for the original system.
In the case of a single-qubit channel, a unitary operation on three qubits would be required.
The best known technique to implement
a general unitary on three qubits
requires a complicated circuit with 20 \textsc{cnot} gates~\cite{MV06}.
An alternative technique~\cite{LV01} uses one ancilla qubit, but
uses a sequence of a large number of interactions, which would require a large number of \textsc{cnot}s.

It is possible to achieve channels far more easily in special cases, or probabilistically,
and to date experimental realizations have had these limitations \cite{Alm07,Qing07,Hanne09,Lee11,FPK+12}.
In particular, a \emph{unital} qubit channel, such as the phase damping channel, can be achieved relatively easily
by applying a random unitary operation.
Alternatively, if one is willing to accept a significant probability of failure, then it is straightforward to provide a method to
generate arbitrary channels~\cite{Piani11}.
In contrast, our technique for qubit channels is general, deterministic, and only requires one \textsc{cnot} and ancilla together with local operations.
As it is already possible to demonstrate a single \textsc{cnot} in several physical systems~\cite{Ladd10},
our technique is implementable with current technology.

We quantify the error tolerance~$\epsilon$
by the Schatten one-norm distance between the simulated channel and the correct channel~\cite{Wat05,KBG+11}.
The classical and quantum algorithms we derive for single-qubit channel simulation
are efficient in that their time and space costs are no worse than polylog$(1/\epsilon)$.
Our algorithms and complexity results for channel simulation rely on decomposing the channel into a
convex combination of simpler channels, dilating each of these channels to unitary mappings
on two qubits~\cite{LV01}, and making use of the Solovay-Kitaev Dawson-Nielsen (SKDN) algorithm~\cite{DN06}.

A succinct statement of the problem we solve follows.
\newtheorem*{theorem}{Problem}
\begin{theorem}
	Construct an efficient autonomous algorithm for designing an efficient quantum circuit,
	implemented from a small single-qubit universal gate set,
	that accurately simulates any completely positive trace-preserving single-qubit mapping
	for any input state within prespecified error tolerance~$\epsilon$
	quantifying the distance between true and approximated states.
\end{theorem}

\noindent
Our solution has the following components:
(i)~the decomposition of arbitrary single-qubit channels as a convex combination of quasiextreme single-qubit channels~\cite{RSW02},
(ii)~a cost reduction of single-qubit channel simulation from requiring a unitary operation on three qubits
	to a circuit with one ancillary qubit and one \textsc{cnot},
(iii)~a geometric lookup database for implementing
	the SKDN algorithm~\cite{DN06}
	to decompose unitary operators, and
(iv)~a proof of efficient simulation by showing that the costs for both the
	classical algorithm for designing the circuit and the quantum circuit itself are at most polylog$(1/\epsilon)$.

Now, we proceed to the technical aspects.
The system is a single qubit
whose state is a positive semidefinite operator $\rho\in\mathcal{T}\left(\mathscr{H}^\text{S}\right)$
with~$\mathscr{H}^\text{S}$ the two-dimensional Hilbert space for the system
and $\mathcal{T}(\mathscr{H})$ denoting
the set of operators on Hilbert space $\mathscr{H}$.
The channel is
\begin{equation}
\label{eq:channel}
	\mathcal{E}:\mathcal{T}(\mathscr{H}^\text{S})\rightarrow\mathcal{T}(\mathscr{H}^\text{S}):
	\rho\mapsto\sum_i K_i \rho K_i^{\dagger},
\end{equation}
with the summation at the end showing the operator-sum representation~\cite{Bha03,Kra83}.
The operators~$\{K_i\}$ are called Kraus operators and satisfy $\sum_iK_i^{\dagger}K_i=\openone$.

The channel can be dilated to a unitary operator on the joint Hilbert space
$\mathscr{H}^\text{SE}=\mathscr{H}^\text{S}\otimes\mathscr{H}^\text{E}$
with~$E$ denoting the environment (or ancillary space) being introduced to purify the dynamics.
Conversion of channel~$\mathcal{E}$ to a Hamiltonian-generated unitary evolution can be achieved
by performing a Stinespring dilation with unitary operator
$U:\mathscr{H}^\text{SE}\rightarrow\mathscr{H}^\text{SE}$, and
\begin{equation}
	\mathcal{U}:
		\mathcal{T}(\mathscr{H}^\text{SE})\rightarrow\mathcal{T}(\mathscr{H}^\text{SE}):
		\rho^\text{SE}\mapsto\rho^{\prime\text{SE}}=U\rho^\text{SE}U^{\dagger}
\end{equation}
such that tr$_\text{E}\rho^\text{SE}=\rho^\text{S}$,
tr$_\text{E}\rho^{\prime\text{SE}}=\rho^{\prime\text{S}}$
and $\mathcal{E}:\rho^\text{S}\mapsto\rho^{\prime\text{S}}$.
Specifically the Kraus operators~(\ref{eq:channel}) have representation
$K_i= {}^\text{E}\!\langle i|U\ket{0}^\text{E}$
for $\ket{i}^\text{E}$ (including $\ket{0}^\text{E}$)
an orthonormal basis state of the environment~\cite{Kra83}.

The unitary operator~$U$ is a minimal dilation of $\mathcal{E}$ if~$U$ is a dilation such that
$\text{dim}\mathscr{H^\text{E}}=\left(\text{dim}\mathscr{H^\text{S}}\right)^2$.
For the case of a single qubit, $\text{dim}\mathscr{H}^\text{E}=4$ for minimal dilation.
Although $\mathscr{H}^\text{E}$ should have a dimension that is the square of the dimension of~$\mathscr{H}^\text{S}$,
and hence four dimensional, we will show that we only require a single resettable ancillary qubit,
so ${\rm dim}\mathscr{H}^\text{E}=2$.

We develop the algorithm for a general single-qubit completely positive trace-preserving (CPTP) map
using the geometrical state representation
$\rho=\frac{1}{2}[\openone+\bm{b}\cdot\bm{\sigma}]$,
where $\bm{b}$ is a three-dimensional vector and
$\bm{\sigma}:=(X,Y,Z)$.
The CPTP map can then be represented by a $4 \times 4$ matrix~\cite{RSW02,WC08}
\begin{equation}
\label{eq:T}
	\mathcal{E}\to\mathbb{T}
		=\begin{pmatrix}
			1&\bm{0}\\ \bm{t}&\bm{T}\\
		\end{pmatrix},\;
		\mathbb{T}_{ij}
			=\frac{1}{2}\text{tr}\left[\sigma_i\mathcal{E}(\sigma_j)\right],\;
		\sigma_0:=\openone,
\end{equation}
with~$\mathbb T$ having 12 independent parameters.
In this representation, the channel is an affine map ~\cite{KR01}
\begin{equation}
	\mathcal{E}:\rho\mapsto\frac{1}{2}(\openone+\bm{b}'\cdot\bm{\sigma}), \quad
	\bm{b}'=\bm{T}\bm{b}+\bm{t}.
\end{equation}
Geometrically, $\mathcal{E}$ maps the state ball into an ellipsoid,
with $\bm{t}$ the shift from the ball's origin
and $\bm{T}$ a distortion matrix for the ball.

In our approach, the channel is constructed from two simpler channels, each of which can be simulated using only one ancillary qubit.
Any single-qubit channel can be decomposed into a convex combination of two channels belonging
to the closure of the set of extreme points of the set of single-qubit channels~\cite{RSW02}.
It turns out that these quasiextreme channels, denoted as~$\mathcal{E}^\text{qe}$, can be simulated using only one ancillary qubit.
In addition, the convex combination is easy to implement, simply by probabilistically implementing one or the other of the quantum channels.

For any CPTP map, the distortion matrix can be transferred into a diagonal form via a singular-value decomposition, so
$\mathcal{E}=\mathcal{U}(\bm{\varphi})\mathcal{E}'\mathcal{U}(\bm{\delta})$
for some~$\mathcal{E}'$ with a diagonal $\bm{T'}$~\cite{KR01}.
In the case of the quasiextreme channel, the shift vector and distortion matrix are of the form~\cite{RSW02}
\begin{align}
\bm{t}'_\text{qe}&=(0,0,\sin\mu \sin\nu)^T, \\
\bm{T}'_\text{qe}&=\text{diag}\left(\cos\nu,\cos\mu,\cos\mu\cos\nu\right)
\end{align}
for some $\mu$ and $\nu$.
This map can be obtained via two Kraus operators
\begin{equation}
	K_0=\begin{pmatrix}\cos\beta &0\\0&\cos\alpha\end{pmatrix}, \quad
	K_1=\begin{pmatrix}0&\sin\alpha\\\sin\beta&0\end{pmatrix},
\end{equation}
where $\alpha=(\mu+\nu)/2$ and $\beta=(\mu-\nu)/2$.
The channel $\mathcal{E}^\text{qe}$ is a generalization of the amplitude damping channel.

\begin{figure}[!b]
\includegraphics[scale=0.2]{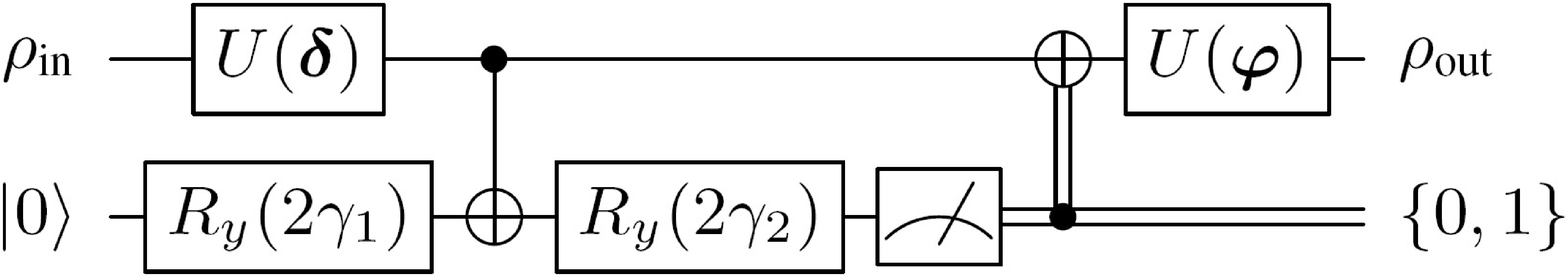}
\caption{The circuit to implement the quasiextreme channel $\mathcal{E}^\text{qe}$.
The unitary operators $U(\bm{\delta})$ and $U(\bm{\varphi})$ serve to diagonalize the channel.}
\label{fig2}
\end{figure}

The circuit to implement the channel $\mathcal{E}^\text{qe}$ is depicted in Fig.~\ref{fig2}.
The rotation takes the form $R_y(2\gamma)=\exp(-iY\gamma)=\openone\cos\gamma-i Y\sin\gamma$;
the two angles in the rotations are $2\gamma_1=\beta-\alpha+\pi/2$ and $2\gamma_2=\beta+\alpha-\pi/2$.
The measurement in the computational basis with the outcome $\ket{0}$ ($\ket{1}$) corresponds to the realization of the Kraus operator $K_0$ ($K_1$).
There is only one \textsc{cnot} required because the final operation is just a classically controlled $X$ operation.

To explain the action of this circuit, note first that the unitary operations $U(\bm{\delta})$ and $U(\bm{\varphi})$
are just the unitaries to diagonalize the distortion matrix.
If the system qubit were in the state $\ket{0}$, then the \textsc{cnot} would have no action on the ancilla, and the two rotations combine to give $R_y(2\beta)$,
which yields the state $\cos\beta\ket{0}+\sin\beta\ket{1}$.
If the system is in the state $\ket{1}$, then an $X$ operation flips the ancilla, and then the two rotations give
$\cos\alpha\ket{0}+\sin\alpha\ket{1}$.
Measuring the ancilla in the state $\ket{0}$ then multiplies state $\ket{0}$ for the system by $\cos\beta$
and state $\ket{1}$ by $\cos\alpha$; this is the action of $K_0$.
Similarly, measuring the ancilla in the state $\ket{1}$ multiplies state $\ket{0}$ for the system by $\sin\beta$,
and state $\ket{1}$ by $\sin\alpha$; this is the action of the operator
\begin{equation}
	K'_1=\begin{pmatrix}\sin\beta &0\\0&\sin\alpha\end{pmatrix}.
\end{equation}
In that case we can simply apply $X$, which gives the required Kraus operator $K_1$.

In contrast, the direct approach to simulate a single-qubit channel is to use Stinespring dilation to construct a unitary acting on the system qubit and two ancillary qubits.
This approach is somewhat inefficient, as a large number of gates is needed to implement a three-qubit unitary.
The best known technique to achieve a three-qubit unitary uses 20 \textsc{cnot}s~\cite{MV06}, although the proven lower bound is 14~\cite{SMB04}.
In contrast, our technique succeeds with only one such gate.
Our result is now summarized in Proposition~\ref{prop:E}.

\begin{Proposition}
\label{prop:E}
Any single-qubit CPTP channel~$\mathcal{E}$
can be simulated with one ancillary qubit, one \textsc{cnot} and four single-qubit operations.
\end{Proposition}
\begin{proof}
From Theorem $14$ in Ref.~\cite{RSW02}, any single-qubit channel~$\mathcal{E}$
can be decomposed into the convex combination $\mathcal{E}=p \mathcal{E}^\text{qe}_1+(1-p)\mathcal{E}^\text{qe}_2$, with $0\leq p \leq1$.
Note that channels $\mathcal{E}^\text{qe}_1$ and $\mathcal{E}^\text{qe}_2$ can be diagonalized,
but the unitary operators to do so may be different in the two cases.
The quasiextreme channels $\mathcal{E}^\text{qe}_i$ can be realized by using the appropriate initial unitary operator, then applying
the circuit above with corresponding angles $\alpha_i$ and $\beta_i$,
and then applying the final unitary operator.
Then, the channel $\mathcal{E}$ can be simulated by randomly implementing the two quasiextreme channels
according to a classical random number generator with probabilities $p$ and $1-p$.
The above circuit uses one \textsc{cnot}, two rotations, a classically controlled $X$ gate,
and two additional unitary operators to diagonalize $\mathcal{E}^\text{qe}_i$.
The final diagonalizing unitary $U(\bm{\varphi})$ may be combined with the $X$ gate,
so only four single-qubit unitary operators are needed.
\end{proof}

In order to complete the decomposition of the channel into a universal gate set, it is necessary to decompose the
single-qubit unitary operators in Proposition~\ref{prop:E} into the gate set.
In the case that the gate set includes Clifford and $T$ ($T=Z^{1/4}$) operations, then any of the techniques
given in Refs.~\cite{DN06,KMM13a,KMM13b,TvMH13,DCS12,PS12} can be used.
Here, we are concerned with the more general problem of what can be achieved with \textsc{cnot}s and a universal single-qubit gate set $S$.
This problem is relevant to experimental situations where not all single-qubit gates can be applied.
The motivation for considering Clifford and $T$ operations in other work is that they are important for encoded logical qubits with error correction,
but such an experiment would be beyond current technology.

We therefore
consider a variation of the SKDN approach \cite{DN06} with the \textsc{cnot} and gates from $S$.
Figure~\ref{fig:skdn}(a) depicts the SKDN strategy by which any single-qubit unitary operator
\begin{figure}[!b]
\includegraphics[scale=0.2]{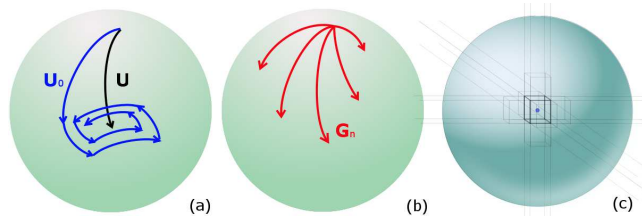}
\caption{
	(color online).
	Schematic diagram for the SKDN algorithm.
	(a) Representation of the algorithm on the Bloch ball.
	The SKDN algorithm finds a polynomial-length gate sequence
	to approximate an arbitrary single-qubit unitary operator~$U$
	by starting with an initial approximation $U_0$ with initial error bound~$\epsilon_0<1/32$
	followed by iteratively constructing operators
	to produce decreasing errors.
	(b)~Action of gates $G_n$ in the lookup database represented as rotations of the sphere.
	(c)~The radius $\pi/2$ ball of single-qubit unitary operations (note that this is different than the Bloch ball).
	Each lookup database gate element~$G_n$ is located within one cube of a
	period $\frac{1}{32\sqrt{3}}$ cubic lattice.
	At the boundary, a cube's center can lie outside the ball but still must be a legitimate domain for the search algorithm.
}
\label{fig:skdn}
\end{figure}
$U(\bm{\theta})=\text{e}^{i\theta_0}\exp(-i\bm{\theta}\cdot\bm{\sigma})$
can be approximately (within~$\epsilon$) decomposed into a
unitary operator~$\tilde{U}=\cdots U_2U_1U_0$ comprising a polylog$(1/\epsilon)$
sequence of gates from $S$~\cite{NC00,KSV02}.

The SKDN algorithm provides an explicit construction
that requires at most $O(\log^{2.71}(1/\epsilon))$ time
and $O(\log^{3.97}(1/\epsilon))$ gates~\cite{DN06}
but requires a database of single-qubit gates $\{G_n\}$
depicted schematically in Fig.~\ref{fig:skdn}(b).
This database gives each~$G_n$ as a sequence of gates from $S$.
However, Dawson and Nielsen do not discuss
how to search this database~\cite{DN06};
we explicitly provide an efficient geometric search technique,
depicted in Fig.~\ref{fig:skdn}(c) and described below.

Ignoring the global phase,
each~$U$ can be identified with coordinate~$\bm{\theta}\in\mathbb{R}^3$.
As $U(\bm{\theta})=U(\bm{\theta}(1-\pi/|\bm{\theta}|))$,
the space~$\mathbb{R}^3$ can be reduced to a radius $\pi/2$ ball,
as depicted in Fig.~\ref{fig:skdn}(c).
We therefore embed a cubic lattice into~$\mathbb{R}^3$ to use as a lookup table.
That is, we construct a database such that, for each cube, there is a sequence of gates that produces a unitary operation within that cube.
Then, if we require a sequence of operators to approximate a given unitary operator,
we identify which cube in the lattice this unitary operator occupies
and then select the corresponding sequence of operators from the database.
Each cube has side length~$\frac{1}{32\sqrt{3}}$, thereby ensuring a maximum separation of~$1/32$ between the unitary and the approximating sequence, which is sufficient for the SKDN algorithm.
For the example of $T$ and $H$ (Hadamard) gates, we find that no more than 36 are required.
(An alternative database lookup procedure is given in Ref.~\cite{TvMH13}.)

Using this database construction with the SKDN algorithm and Proposition~\ref{prop:E}, we have an explicit algorithm
to decompose a single-qubit channel into \textsc{cnot}s and gates from $S$.
This classical design algorithm accepts as input the
error tolerance $\epsilon$ for the single-qubit channel and the channel parameters~$\mathcal E$.
As output, the algorithm delivers the description
of the quantum algorithm implemented as a sequence of gates from the instruction set.

The procedure to be followed will depend on what single-qubit gates are available experimentally and the desired accuracy.
For experiments in the near future, the best approach is likely to be to simply use a lookup database directly,
as it will be challenging to produce single-qubit sequences longer than 36.
Alternatively, if the full set of single-qubit unitaries is available, then one may use the circuit in Fig.~\ref{fig2} directly.
The procedure outlined above should be used if there is a restricted single-qubit gate set and high precision is required.
In the special case that the single-qubit gate set is $\{H,T\}$, then one can use new techniques such as those in Ref.~\cite{VK13}.

For completeness, we need to bound the error in the channel in terms of the error in the unitary.
For the unitary,
the error is simply the worst-case two-norm distance between the true and approximate pure states
in the system Hilbert space
\begin{equation}
	\|U-\tilde{U}\|:=\max_{\ket{\psi}} \|(U-\tilde{U})\ket{\psi}\|.
\end{equation}
The appropriate measure of error for the channel
 is the Schatten one-norm~\cite{Wat05,KBG+11}
\begin{equation}
\label{eq:error}
	\|\mathcal{E}-\mathcal{\tilde{E}}\|_{1\rightarrow1}
		:=\max_{\rho}\|\mathcal{E}(\rho)-\mathcal{\tilde{E}}(\rho)\|_{\text{1}},\;
	\|\bullet\|_{\text{1}} :=\text{tr}\sqrt{\bullet^\dagger\bullet}.
\end{equation}
The following proposition establishes that the channel-simulation error condition is satisfied
if the error bound for the dilated unitary operator~$U$ is~$\epsilon/2$.
\begin{Proposition}\label{prop:pro1}
	For CPTP maps
	$\mathcal{E}, \mathcal{\tilde{E}}:\mathcal{T}(\mathscr{H})\rightarrow\mathcal{T}(\mathscr{H})$
	with respective minimal dilations $U, \tilde{U}$: $\mathscr{H}\otimes\mathscr{H}'\rightarrow\mathscr{H}\otimes\mathscr{H}'$,
	then $\|\mathcal{E}-\mathcal{\tilde{E}}\|_{1\rightarrow 1}\leq 2\|U-\tilde{U}\|$.
\end{Proposition}
\begin{proof}
Using Eq.~(17) of Ref.~\cite{HLSW04} and the convexity of trace distance,
\begin{align}
2\max_{\ket{\psi}}\|(U-\tilde{U})\ket{\psi}\| &\ge \max_{\ket{\psi}}\|U\ket{\psi}\langle\psi|U^\dagger-\tilde{U}\ket{\psi}\langle\psi|\tilde{U}^\dagger\|_1 \nonumber \\
&\ge \max_{\rho}\|\mathcal{E}(\rho)-\mathcal{\tilde{E}}(\rho)\|_{\text{1}}.
\end{align}
Using the definitions, this immediately gives the required inequality.
\end{proof}

We now articulate our complete result for the decomposition
of the channel into the universal gate set.
\begin{Proposition}
\label{prop:final}
Any single-qubit channel~$\mathcal{E}$ can be approximated within
one-norm distance $\epsilon$ using $O(\log^{3.97}(1/\epsilon))$ computer time and gates from the set $S$,
and using one \textsc{cnot}, one ancillary qubit and one classical bit.
\end{Proposition}
\begin{proof}
First, via Proposition~\ref{prop:E} the channel can be decomposed into a convex combination of channels,
and thereby simulated using one ancilla qubit, one \textsc{cnot} operation, and four single-qubit unitary operators.
Provided each of the channels in the convex combination is simulated within distance $\epsilon$, the
overall channel is simulated within distance $\epsilon$ by the convexity of the one-norm distance.

Via Proposition~\ref{prop:pro1}, the error bound for the channel is satisfied if the two-qubit unitary operators are
approximated within distance $\epsilon/2$.
There are four single-qubit unitary operators used within the circuit.
The error bound will be satisfied, provided each of these unitary operators is approximated within distance
$\epsilon/8$.
These unitary operators can be approximated via the SKDN algorithm with $O(\log^{3.97}(1/\epsilon))$ gates from $S$.
Using our lookup database, the SKDN algorithm may be implemented efficiently, in that the classical complexity
to determine the gate sequence does not exceed $O(\log^{2.71}(1/\epsilon))$.
\end{proof}

We now have the full algorithm for open-system single-qubit channel quantum simulation.
For a given input channel, the channel will be decomposed into the form in Proposition~\ref{prop:E}, and then the single-qubit rotations which contain continuous variables therein will further be decomposed into sequences of universal gates satisfying the error condition.
This simulator accepts the initial state~$\rho$ and yields the approximate output state $\mathcal{\tilde E}(\rho)$
while satisfying the error condition of Proposition~\ref{prop:pro1}.

This scheme could be implemented in a number of quantum computing architectures.
For example, it could be implemented with linear optics, although in that case, the \textsc{cnot} is nondeterministic,
and other methods are available to perform nondeterministic channels \cite{Piani11,FPK+12}.
A promising architecture to deterministically demonstrate this scheme is trapped ions.
\textsc{cnot} gates have been demonstrated with error below $0.01$~\cite{BKRB04}, and
single-qubit gates have been demonstrated with error below $10^{-4}$~\cite{BWC+11}.
In the case of trapped ions, it is possible to perform general single-qubit gates, so it is not necessary to use gate sequences.
Nevertheless, the ability to perform large numbers of sequential single-qubit operations (nearly $1000$ in Ref.~\cite{BWC+11})
means that gate sequences could easily be demonstrated.

In summary, we have shown how to implement a single-qubit channel using the \textsc{cnot} and a universal set of single-qubit gates $S$.
This can be regarded as a quantum simulation, except it differs from other quantum simulation methods in
that we directly simulate the mapping
rather than continuous-time evolution.
Our quantum circuit is appealing for experimental implementation because only two qubits are necessary,
rather than three as the Stinespring dilation theorem suggests.
As a result, only one \textsc{cnot} operation is needed, as compared to 20 for a straightforward application of Stinespring dilation.
When decomposing the single-qubit unitary operators into gates from $S$,
the number of gates and classical complexity follow from the Solovay-Kitaev Dawson-Nielsen algorithm.
This work raises a number of questions for future research.
Most importantly, is it possible to achieve similar simplifications for \emph{qudit} channels?
Another question is whether it is possible to obtain further simplifications for the simulation of qubit channels.

\begin{acknowledgments}
	We thank J.\ Eisert, C.\ Horsman, P.\ H{\o}yer, and N.\ Wiebe for valuable discussions.
	D.S.W. acknowledges financial support from USARO.
	Subsequent to the appearance of our work on the arXiv,
	S. D. Bartlett sent us his related unpublished work,
	which showed that the quasiextreme channel in Fig.~1 required only one \textsc{cnot}
	gate and not two~\cite{BBS--}, and we modified accordingly.
	We thank an anonymous referee for the simplified proof of Proposition~\ref{prop:pro1}.
	M.C.O. acknowledges support from AITF
	and the Brazilian agencies CNPq and FAPESP
	through the Instituto Nacional de Ci{\^e}ncia e Tecnologia -- Informa\c{c}{\~a}o Qu{\^a}ntica (INCT-IQ).
	D.W.B. is funded by an ARC Future Fellowship (FT100100761).
	B.C.S. acknowledges AITF, CIFAR, NSERC and USARO for financial support.
	This project was supported in part by the National Science Foundation under Grant No.\ NSF PHY11-25915.
\end{acknowledgments}
\bibliography{oss}
\end{document}